\documentclass[11pt,letterpaper]{article}
	\usepackage{libertine}
	\usepackage{fullpage}
	\usepackage[breaklinks]{hyperref}
	\hypersetup{colorlinks={true},urlcolor={blue},linkcolor={DarkBlue},citecolor=[named]{DarkGreen}}
	\usepackage[authoryear,square]{natbib}
	\usepackage[svgnames]{xcolor}
	\renewcommand{\cite}{\citep}
	\usepackage{authblk}
    \title{\bf The Metric Distortion of Multiwinner Voting}
    
    \author[1]{Ioannis Caragiannis}
    \author[2]{Nisarg Shah}
    \author[3]{Alexandros A. Voudouris}
    
    \affil[1]{Department of Computer Science, Aarhus University}
    \affil[2]{Department of Computer Science, University of Toronto}
    \affil[3]{School of Computer Science and Electronic Engineering, University of Essex}
	\date{}

\usepackage{amsmath,amsfonts,amssymb,amsthm,amsbsy,bbm}
\usepackage{tikz,tikz-qtree}
\usetikzlibrary{arrows}
\usepackage{pgfplots}
\usepackage{graphicx}
\usepackage{subcaption}
\usepackage[capitalise,nameinlink]{cleveref}
\usepackage{tikz}  
\usetikzlibrary{math}
\usetikzlibrary{arrows}
\usetikzlibrary{patterns,snakes}
\usetikzlibrary{decorations.shapes}
\tikzstyle{overbrace text style}=[font=\tiny, above, pos=.5, yshift=5pt]
\tikzstyle{overbrace style}=[decorate,decoration={brace,raise=5pt,amplitude=3pt}]
\usetikzlibrary{shapes.geometric}

\usepackage{multirow}

\usepackage[linesnumbered,ruled,vlined]{algorithm2e}
\usepackage{xspace}
\crefname{algocf}{alg.}{algs.}
\Crefname{algocf}{Algorithm}{Algorithms}

\allowdisplaybreaks

\newtheorem{theorem}{Theorem}
\newtheorem{corollary}{Corollary}
\newtheorem{lemma}{Lemma}
\newtheorem*{claim*}{Claim}
\newtheorem*{theorem*}{Theorem}

\theoremstyle{definition}

\DeclareMathOperator*{\argmin}{\arg\min}

\newcommand{\floor}[1]{\lfloor{#1}\rfloor}
\newcommand{\ceil}[1]{\lceil{#1}\rceil}
\renewcommand{\top}{\operatorname{top}}

\newcommand{\set}[1]{\{#1\}}
\renewcommand{\subset}{\subseteq}
\renewcommand{\ge}{\geqslant}
\renewcommand{\geq}{\geqslant}
\renewcommand{\le}{\leqslant}
\renewcommand{\leq}{\leqslant}

\newcommand{\bell}{\boldsymbol \ell}

\newcommand{\SC}{\text{SC}}

\newcommand{\E}{\mathbb{E}}
\newcommand{\N}{\mathbb{N}}

\newcommand{\types}{\mathcal{T}}

\begin{document}

\maketitle

\begin{abstract}
We extend the recently introduced framework of metric distortion to multiwinner voting. In this framework, $n$ agents and $m$ alternatives are located in an underlying metric space. The exact distances between agents and alternatives are unknown. Instead, each agent provides a ranking of the alternatives, ordered from the closest to the farthest. Typically, the goal is to select a single alternative that approximately minimizes the total distance from the agents, and the worst-case approximation ratio is termed distortion. In the case of multiwinner voting, the goal is to select a committee of $k$ alternatives that (approximately) minimizes the total cost to all agents. We consider the scenario where the cost of an agent for a committee is her distance from the $q$-th closest alternative in the committee. We reveal a surprising trichotomy on the distortion of multiwinner voting rules in terms of $k$ and $q$: The distortion is unbounded when $q \leq k/3$, asymptotically linear in the number of agents when $k/3 < q \leq k/2$, and constant when $q > k/2$.
\end{abstract}

\section{Introduction}\label{sec:intro}

The most canonical problem in voting theory is to aggregate ranked preferences of $n$ individual agents over a set of $m$ alternatives to reach a collective decision. Examples of such decisions include selecting a single alternative (single-winner voting), selecting $k$ out of $m$ alternatives for a fixed $k > 1$ (multiwinner voting), and selecting a set of costly alternatives subject to a budget constraint (participatory budgeting). In centuries of research on voting, perhaps the most prominent approach to designing voting rules has been the axiomatic approach, in which one fixes several qualitative axioms and seeks voting rules satisfying them. Unfortunately, this approach has often led to impossibility results, such as Arrow's impossibility theorem~\cite{Arr51}.

To circumvent this, \citet{procaccia2006distortion} proposed the (utilitarian) \emph{distortion} framework for analyzing, comparing, and designing single-winner voting rules. Under this framework, the ordinal preferences expressed by agents are viewed as proxies for their underlying cardinal utilities, and the goal of a voting rule is to optimize the worst-case approximation ratio (distortion) of the social welfare (the total utility of the agents). Under minimal assumptions, this framework offers a \emph{quantitative} comparison of voting rules. It has been used successfully to analyze the distortion of well-known methods~\cite{caragiannis2011embedding} and to identify voting rules with optimal distortion~\cite{boutilier2015optimal}. The framework has also been extended to multiwinner voting~\cite{caragiannis2017subset} and participatory budgeting~\cite{benade2017participatory}, under the assumption that the utility of an agent for a set of alternatives is the maximum and the sum of her utilities for the alternatives in the set, respectively. 

\citet{anshelevich2015approximating} built on this idea to propose the \emph{metric distortion} framework, in which agents and alternatives are embedded in an underlying metric space, and the \emph{cost} of an agent for an alternative is the distance between them. An agent still ranks the alternatives, but now in a non-decreasing order of their distance from her. Instead of maximizing the social welfare, the goal is now to minimize the social cost (the total cost of the agents). Like in the utilitarian case, scholars have analyzed the metric distortion of prominent voting rules~\cite{skowron2017stv,goel2017metric,munagala2019improved,kempe2020duality} and identified rules with optimal metric distortion~\cite{gkatzelis2020resolving}. For a detailed overview, we refer the reader to the survey by \citet{anshelevich2021distortion}. 

The goal of our work is to extend the metric distortion framework to multiwinner voting, where the objective is to select a subset of $k$ alternatives (committee) for a given $k > 1$. To the best of our knowledge, the only prior works to address multiwinner voting with general metric costs are those of \citet{goel2018relating} and \citet{chen2020favorite}. Both focus on a model in which an agent's cost for a committee is the sum of her distances to the alternatives in the committee. In a sense, this assumes that an agent cares equally about all the alternatives chosen in the committee. 

However, in many applications, an agent may care only about a few alternatives in the committee, typically the ones she prefers more. For example, when parliament members are chosen in a political election, each voter may associate with just one (or a few) of the elected candidates as her representative(s). Similarly, if a city builds parks at multiple locations, each resident may only be able to access a few parks closest to her. Motivated by these applications, we consider the case where an agent's cost for a committee of $k$ alternatives is her distance to the $q$-th closest alternative in the committee, for a given $q \le k$. Note that for $q=1$, this is the distance to the closest alternative in the committee, whereas for $q=k$, it is the distance to the farthest one. Our main research question is:
\begin{quote}
	\emph{What is the optimal distortion for selecting a committee of $k$ alternatives under this cost model? How does it depend on the relation between $q$ and $k$? Can the optimal distortion be achieved via efficient voting rules?}
\end{quote} 

Before proceeding further, note that another reason why it may not be desirable to model an agent's cost for a committee as the \emph{sum} of her distances to the alternatives in the committee is that the optimal committee can suffer from the tyranny of the majority; that is, it may consist entirely of the alternatives liked by the majority and include none liked by a minority. This is somewhat reflected by the fact that repeatedly applying a single-winner voting rule is known to work well in this model~\cite{goel2018relating}. Due to this, prior work on multiwinner voting, both in the distortion literature~\cite{caragiannis2017subset} and elsewhere~\cite{CC83,Mon95,PRZ08,LB11c}, aims to ensure that there is \emph{at least one} alternative in the committee that every agent likes well (which corresponds to $q=1$ in our model). An interesting byproduct of our cost model is that an agent's submitted ranking of individual alternatives can be used to determine her ordinal preferences over committees; we use this fact to derive some of our results. 

\subsection{Our Contributions}\label{sec:contributions}
Recall that $n$ is the number of agents and $m$ is the number of alternatives. We reveal a surprising trichotomy on the best possible distortion for multiwinner voting. 

When $q \le k/3$, the distortion is unbounded. This holds even for randomized voting rules, and even when $n$ and $m$ are only linear in $k$. When $q \in (k/3,k/2]$, the best possible distortion is $\Theta(n)$. Here, the upper bound is obtained via a novel voting rule that is deterministic and efficient, while the lower bound holds even for randomized voting rules and when $m$ is linear in $k$ or $q$. Finally, when $q > k/2$, the best possible distortion is $3$ for deterministic rules, and between $2$ and $3$ for randomized rules. For this case, we show that the costs of agents for committees satisfy the triangle inequality. Hence, we can reduce multiwinner voting to single-winner voting by treating each committee as a separate alternative and applying any single-winner voting rule, allowing us to borrow known best possible distortion bounds from the single-winner case to the multiwinner one. This is where we use the fact that under our cost model, we can deduce an agent's ordinal preferences over committees from her provided ranking of individual alternatives. 

However, this reduction is not efficient as we need to apply the single-winner voting rule on an instance with $\binom{m}{k}$ alternatives. To that end, we also provide an efficient reduction. We show that there exists an agent such that the committee consisting of the $k$ alternatives she prefers the most has social cost no worse than $3$ times the optimal. Thus, we can apply any single-winner voting rule with distortion $\rho$ on a reduced instance with only $n$ committees (one committee per agent) as alternatives and obtain a multiwinner rule with distortion at most $3\rho$. In particular, by applying the \textsc{PluralityMatching} rule of \citet{gkatzelis2020resolving}, which is known to achieve the best possible distortion of $3$ in the single-winner case, we obtain distortion at most $9$ in polynomial time. We also show that an efficient reduction of this type cannot be used to achieve distortion better than $5.207$. An overview of our results is given in Table~\ref{tab:results}.

\begin{table}[t]
    \centering
    \begin{tabular}{c|cc}
                            & Deterministic 	    & Randomized \\ \hline 
       $q \leq k/3$         &   $+\infty$    		& $+\infty$     \\[8pt]
       $q \in (k/3, k/2]$   &   $\Theta(n)$      	& $\Theta(n)$ 	  \\[8pt]
       \multirow{3}{*}{$q > k/2$}	&  $3$ {\small (exp time, general $n$)}    & \multirow{3}{*}{$3-2/n$}	   \\ 
                            &  $3$ {\small (poly time, constant $n$)} & \\
					        &  $\leq 9$ {\small (poly time, general $n$)} & \\ \hline
    \end{tabular}
    \caption{An overview of our distortion bounds.}
    \label{tab:results}
\end{table}

\subsection{Related Work}\label{sec:related}
The works of \citet{goel2018relating} and \citet{chen2020favorite} are the most closely related to ours as they consider metric distortion for multiwinner voting. As mentioned before, both focus on a setting where the cost of an agent for a committee is the sum of her distances to the alternatives in the committee.\footnote{\citet{chen2020favorite} focus only on the case of $k=m-1$, but also consider another cost model.} \citet{goel2018relating} show that selecting a committee by repeatedly applying any single-winner voting rule with distortion $\delta$ yields a distortion of at most $\delta$. Since the best possible single-winner distortion is known to be $3$~\cite{gkatzelis2020resolving}, this implies that a distortion of $3$ can be achieved for any $k$ in this model. \citet{chen2020favorite} directly present a voting rule achieving distortion $3$ for the case of $k=m-1$ and prove this to be best possible. 

\citet{jaworski2020evaluating} consider a model where voters (agents) and candidates (alternatives) have preferences over a set of \emph{binary issues}, which induce the preferences of the voters over the candidates; specifically, voters rank candidates based on the number of issues they agree on. The elected committee uses majority voting to decide on each issue, and the cost of a voter is the number of issues for which the decision differs from her preferred outcome. This can be viewed as a metric distortion model, but with the specific Hamming distance metric. Also, the cost of a voter for a committee depends on the locations of the candidates in the committee, and not just on their distances to her. \citet{meir2021representative} consider a similar model where voters are also candidates, and show that sortition --- picking $k$ of the voters uniformly at random --- leads to low distortion. 

The distortion of randomized rules has also received significant attention. For single-winner voting, the best possible distortion under the utilitarian model is known to be $\tilde{\Theta}(\sqrt{m})$~\cite{boutilier2015optimal}, while that under the metric model still remains a challenging open question~\cite{anshelevich2017randomized,kempe2020communication,charikar2022randomized}. \citet{caragiannis2017subset} provide bounds on the best possible distortion of randomized multiwinner voting rules under the utilitarian model, and our work does so under the metric model. 

More broadly, there is a huge literature on multiwinner voting that focuses on desiderata other than distortion, such as proportional representation~\cite{aziz2017justified,peters2020proportionality}, committee diversity~\cite{bredereck2018multiwinner}, monotonicity or consistency axioms~\cite{elkind2017properties}, explainability~\cite{peters2021market}, etc. We refer the interested reader to the book chapter by \citet{faliszewski2017multiwinner} for an overview. 

\section{Preliminaries}
For $p \in \mathbb{N}$, define $[p] = \set{1,\ldots,p}$. An instance of our problem is given by the tuple $I=(N,A,d,k,q)$, where:
\begin{itemize}
    \item $N$ is a set of $n \geq 2$ {\em agents}.
    \item $A$ is a set of $m \geq 2$ {\em alternatives}.
    \item $d$ is a pseudometric over $N \cup A$ with $d(x,y)$ denoting the \emph{distance} between $x,y \in N \cup A$. Being a pseudometric, $d$ satisfies, for all $x,y,z \in N \cup A$, $d(x,x) = 0$, $d(x,y)=d(y,x)$, and the \emph{triangle inequality} $d(x,y) \le d(x,z)+d(z,y)$. Since our framework only uses distances between agents and alternatives (and not between two agents or between two alternatives), we use the following equivalent formulation of the triangle inequality~\cite{goel2017metric}: $d(i,x) \leq d(i,y) + d(j,y) + d(j,x)$ for all agents $i,j \in N$ and alternatives $x,y \in A$.\footnote{When proving lower bounds, we will often design a worst-case pseudometric $d$ by embedding agents and alternatives in the 1D Euclidean space and taking the Euclidean distance between them.}
    \item $k$ and $q$ are positive integers such that $1 \leq q \leq k < m$. 
\end{itemize}

Every agent $i \in N$ ranks the alternatives based on her distances from them, from smallest (most preferable) to largest (least preferable), breaking ties arbitrarily; that is, the pseudometric $d$ induces the ordinal preferences of agent $i$ given by a {\em ranking} $\succ_i$ over the alternatives such that $x \succ_i y$ implies $d(i,x) \leq d(i,y)$. We refer to $\succ_d = (\succ_i)_{i \in N}$ as the {\em preference profile}. 

For any $S \subseteq A$ with $|S| \ge q$, we denote by $\top_{i,q}(S)$ the set of $q$ most preferred alternatives of agent $i$ in $S$. A {\em committee} $C \subseteq A$ is a set of alternatives of size exactly equal to $k$. We define the {\em $q$-cost} of agent $i$ for $C$, denoted $c_{i,q}(C|d)$, to be the distance of $i$ from her {\em $q$-th closest alternative} in $C$: 
$c_{i,q}(C|d) = \max_{x \in \top_{i,q}(C)} d(i,x)$. 
The {\em $q$-social cost} of committee $C$, denoted $\SC_q(C|d)$, is then defined as the total $q$-cost of the agents for $C$: 
$\SC_q(C|d) = \sum_{i \in N} c_{i,q}(C|d)$.

A (randomized) {\em multiwinner voting rule} $f$ takes as input a preference profile $\succ_d$ and outputs a distribution over committees; we say that $f$ is deterministic if the distribution it returns has singleton support. Given $k \in \mathbb{N}$ and $q \in [k]$, the {\em $(k,q)$-distortion} of $f$ is the worst-case ratio, over all possible instances with these parameters, between the (expected) $q$-social cost of the committee chosen by $f$ and the minimum possible $q$-social cost of any committee, i.e.,
\begin{align*}
\sup_{I = (N,A,d,k,q) : |A|=m>k} \frac{\mathbb{E}[\SC_q(f(\succ_d)|d)]}{\min_{C \subset A: |C|=k} \SC_q(C|d)}. 
\end{align*}
Our goal is to design multiwinner rules with as low $(k,q)$-distortion as possible. 

To simplify notation, we drop $q$ and $d$ whenever they are clear from the context. In particular, we will use $c_i(C)$ instead of $c_{i,q}(C|d)$, $\SC(C)$ instead of $\SC_q(C|d)$, and $\top_i(S)$ instead of $\top_{i,q}(S)$. We also denote by $T_i = \top_i(A)$ the set of $q$ alternatives ranked highest by agent $i$. 

\section{Unbounded Distortion With $q \leq k/3$}
We begin with a strong impossibility result for the case where $q \leq k/3$. In particular, we show that every multiwinner voting rule has unbounded distortion, even if it is allowed to use randomization. 

\begin{theorem} \label{thm:q<=k/3-LB}
For $q \leq k/3$, the $(k,q)$-distortion of every (even randomized) multiwinner voting rule is unbounded.
\end{theorem}
\begin{proof}
Fix $k$ and $q$ such that $q \leq k/3$, and a mutliwinner voting rule $f$. Let $L = \lfloor k/q \rfloor + 1 \ge 4$. We consider an instance with $n=L$ agents, partitioned into two sets: $V=\set{v_1,\dots,v_{\lfloor L/2 \rfloor}}$ and $U=\{u_1,\ldots,u_{\lceil L/2 \rceil}\}$. There are $m=Lq$ alternatives, partitioned into $L$ sets $X_1, \ldots, X_{\lfloor L/2 \rfloor},Y_1, \ldots, Y_{\lceil L/2 \rceil}$ of size $q$ each.\footnote{We use an instance with $n \leq m$. Note, however, that the lower bound continues to hold even if one assumes $n \ge m$. We can simply create $t$ copies of each agent, for an appropriately large $t \in \mathbb{N}$.} Let $X = \bigcup_{\ell=1}^{\lfloor L/2 \rfloor} X_\ell$ and $Y = \bigcup_{\ell= 1}^{\lceil L/2 \rceil} Y_\ell$. 

Consider any preference profile such that:
\begin{itemize}
\item
Every agent in $V$ ranks the alternatives in $X$ higher than those in $Y$. 

\item 
Every agent in $U$ ranks the alternatives in $Y$ higher than those in $X$.

\item 
For $\ell \in [\floor{L/2}]$, every agent ranks the alternatives in $X_\ell$ as well as those in $Y_\ell$ arbitrarily. 

\item 
For $\ell \in [\floor{L/2}]$, agent $v_\ell \in V$ ranks the alternatives in $X_i$ higher than those in $X_j$ whenever $|\ell-i| < |\ell-j|$, for $i,j \in [\floor{L/2}]$; agent $v_\ell$ ranks the sets of $Y$ in the order $Y_1,\ldots,Y_{\ceil{L/2}}$ from highest to lowest.

\item 
For $\ell \in [\ceil{L/2}]$, agent $u_\ell \in U$ ranks the alternatives in $Y_i$ higher than those in $Y_j$ whenever $|\ell-i| < |\ell-j|$, for $i,j \in [\ceil{L/2}]$; agent $u_\ell$ ranks the sets of $X$ in the order $X_{\floor{L/2}},\ldots,X_1$ from highest to lowest.
\end{itemize}
Table~\ref{tab:q<=k/3-LB} presents an example preference profile for $k=8$ and $q=2$; Figure~\ref{fig:q<=k/3-LB} depicts the possible underlying metric spaces used in the two cases below for this example instance.

\begin{table}[t]
\centering
\begin{tabular}{c|c}
agent & ranking \\ \hline
$v_1$  &  $X_1 \succ X_2 \succ Y_1 \succ Y_2 \succ Y_3$ \\ 
$v_2$  &  $X_2 \succ X_1 \succ Y_1 \succ Y_2 \succ Y_3$ \\ \hline
$u_1$  &  $Y_1 \succ Y_2 \succ Y_3 \succ X_2 \succ X_1$  \\
$u_2$  &  $Y_2 \succ Y_1 \succ Y_3 \succ X_2 \succ X_1$  \\
$u_3$  &  $Y_3 \succ Y_2 \succ Y_1 \succ X_2 \succ X_1$ 
\end{tabular}
\caption{An example of the preference profile used in the proof of Theorem~\ref{thm:q<=k/3-LB}, when $k=8$ and $q=2$. Here, we have $L = \lfloor 8/2 \rfloor + 1 = 5$, $n=L=5$, and $m = Lq = 10$. The alternatives are partitioned into $5$ sets $X_1,X_2, Y_1, Y_2, Y_3$ of size $2$ each. We also have that $X=\{X_1,X_2\}$ and $Y=\{Y_1,Y_2,Y_3\}$. As an example, $u_2$ ranks the alternatives of $Y$ higher than those of $X$, i.e., $\{Y_1,Y_2,Y_3\} \succ \{X_1,X_2\}$. Within $Y$, the sets therein are ranked based on their index distance from $2$, e.g., $Y_2 \succ Y_1 \succ Y_3$. Within $X$, the order is fixed so that $X_2 \succ X_1$. The two alternatives in every subset are ranked arbitrarily.
}
\label{tab:q<=k/3-LB}
\end{table}

Because $m = Lq = (\floor{k/q}+1)\cdot q > k$, not all alternatives can be included in the committee. We switch between the following two cases.

\begin{figure}[t]
\tikzset{every picture/.style={line width=0.75pt}} 
\centering
\begin{subfigure}[t]{0.45\linewidth}
\centering
\begin{tikzpicture}[x=0.7pt,y=0.7pt,yscale=-1,xscale=1]
\draw [line width=0.75]  (0,0) -- (280,0) ;
\filldraw (20,0) circle (3pt);
\filldraw (100,0) circle (3pt);
\filldraw (180,0) circle (3pt);
\filldraw (260,0) circle (3pt);

\draw (20,-40) node [inner sep=0.75pt]  [font=\small]  {$\{v_1, X_1\}$};
\draw (20,-20) node [inner sep=0.75pt]  [font=\small]  {$\{v_2, X_2\}$};
\draw (20,20) node [inner sep=0.75pt]  [font=\small]  {$0$};

\draw (100,-20) node [inner sep=0.75pt]  [font=\small]  {$\{u_1, Y_1\}$};
\draw (100,20) node [inner sep=0.75pt]  [font=\small]  {$4$};

\draw (180,-20) node [inner sep=0.75pt]  [font=\small]  {$\{u_2 , Y_2\}$};
\draw (180,20) node [inner sep=0.75pt]  [font=\small]  {$5$};

\draw (260,-20) node [inner sep=0.75pt]  [font=\small]  {$\{u_3, Y_3\}$};
\draw (260,20) node [inner sep=0.75pt]  [font=\small]  {$6$};
\end{tikzpicture}
\caption{The metric space in Case 1.}
\end{subfigure}
\ \ \ \ \ 
\begin{subfigure}[t]{0.45\linewidth}
\centering
\begin{tikzpicture}[x=0.7pt,y=0.7pt,yscale=-1,xscale=1]
\draw [line width=0.75]  (0,0) -- (200,0) ;
\filldraw (20,0) circle (3pt);
\filldraw (100,0) circle (3pt);
\filldraw (180,0) circle (3pt);

\draw (20,-20) node [inner sep=0.75pt]  [font=\small]  {$\{v_1, X_1\}$};
\draw (20,20) node [inner sep=0.75pt]  [font=\small]  {$-4$};

\draw (100,-20) node [inner sep=0.75pt]  [font=\small]  {$\{v_2, X_2\}$};
\draw (100,20) node [inner sep=0.75pt]  [font=\small]  {$-3$};

\draw (180,-60) node [inner sep=0.75pt]  [font=\small]  {$\{u_1, Y_1\}$};
\draw (180,-40) node [inner sep=0.75pt]  [font=\small]  {$\{u_2, Y_2\}$};
\draw (180,-20) node [inner sep=0.75pt]  [font=\small]  {$\{u_3, Y_3\}$};
\draw (180,20) node [inner sep=0.75pt]  [font=\small]  {$0$};
\end{tikzpicture}
\caption{The metric space in Case 2.}
\end{subfigure}
\caption{The two metrics used in the proof of Theorem~\ref{thm:q<=k/3-LB} when $k=8$ and $q=2$. Both are consistent with the ordinal preferences presented in Table~\ref{tab:q<=k/3-LB}. If the $6$ alternatives of $Y = Y_1 \cup Y_2 \cup Y_3$ are not all included in the committee with positive probability, the expected $q$-cost of some agent of $U=\{u_1,u_2,u_3\}$ is positive in the first metric, thus leading to unbounded distortion, as the committee that includes the $6$ alternatives of $Y$ and $2$ alternatives of $X=X_1 \cup X_2$ has $q$-social cost $0$. If, on the other hand, all the alternatives of $Y$ are included in the committee, not all $4$ alternatives of $X$ can be included in the committee, and thus some agent in $V=\{v_1,v_2\}$ has $q$-cost at least $1$ in the second metric in any possible scenario; this again leads to unbounded distortion as the committee that includes the $4$ alternatives of $X$ and $4$ alternatives of $Y$ has $q$-social cost $0$.}
\label{fig:q<=k/3-LB}
\end{figure}
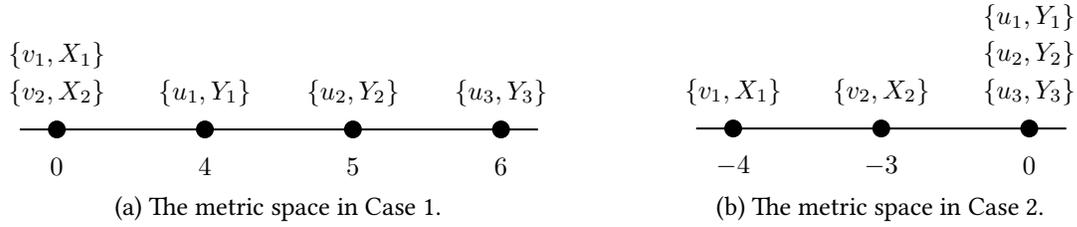

\paragraph{Case 1: Some alternative in $Y$ is not included in the committee with a positive probability.} \ \\ 
Suppose that an alternative in $Y_{\ell^*}$ is not included in the committee with a positive probability, for some $\ell^* \in [\ceil{L/2}]$. Consider the following one-dimensional Euclidean metric, which is consistent with the ordinal preferences of the agents defined above:
\begin{itemize}
\item 
All agents in $V$ and all alternatives in $X$ are located at $0$.

\item 
For every $\ell \in [\ceil{L/2}]$, agent $u_\ell \in U$ and the alternatives in $Y_\ell$ are located at $\lceil L/2 \rceil + \ell$.
\end{itemize}

Since some alternative in $Y_{\ell^*}$ is not included in the committee with a positive probability, the expected $q$-cost of agent $u_{\ell^*}$ is positive, and thus the expected $q$-social cost under $f$ is also positive. However, it is possible to achieve $q$-social cost $0$ by including in the committee all the alternatives of $Y$ and an arbitrary subset of $k - \lceil L/2 \rceil q \geq q$ alternatives from $X$; here, the inequality follows because $q \le k/3$. So, the distortion of $f$ is unbounded in this case.

\paragraph{Case 2: Every alternative in $Y$ is included in the committee with probability $1$.} \ \\
Consider the following metric, which is consistent with the ordinal preferences of the agents defined above:
\begin{itemize}
\item
For $\ell \in [\lfloor L/2 \rfloor]$, agent $v_\ell$ and the alternative in $X_\ell$ are located at $-L + \ell$.

\item 
All agents of $U$ and all alternatives in $Y$ are located at $0$.
\end{itemize}

Since $f$ chooses a committee that includes all of $Y$ with probability $1$, it excludes at least one alternative of $X$ with probability $1$. Whenever an alternative in $X_\ell$ is excluded, note that agent $v_\ell$ has $q$-cost at least $1$. Hence, the expected $q$-social cost is at least $1$. However, it is possible to achieve $q$-social cost $0$ by choosing the committee containing all the alternatives of $X$ and an arbitrary subset of $k-\lfloor L/2 \rfloor q \geq q$ alternatives from $Y$. So, the distortion of $f$ is unbounded in this case as well. 
\end{proof}

\section{Linear Distortion With $k/3 < q \leq k/2$} \label{sec:k/3<q<=k/2}
We now turn our attention to the case of $q \in (k/3, k/2]$. In this case, the best possible distortion turns out to be bounded, but linear in the number of agents, which could be very large.

We propose a novel deterministic multiwinner voting rule, called {\sc PolarOpposites} (see  \Cref{alg:2-agentFavorite}), which runs in polynomial time and achieves a distortion of $O(n)$. Recall that, for a fixed value of $q$, $\top_i(S)$ is the set of the $q$ most favorite alternatives of agent $i$ in $S$, and $T_i = \top_i(A)$. {\sc PolarOpposites} is relatively straightforward: we choose an agent $i$ arbitrarily, and another agent $j$ whose $T_j$ has the highest cost\footnote{Here, we use the fact that in our model, we can compare two sets of alternatives in terms of their cost to agent $i$ using only agent $i$'s preference ranking over individual alternatives.} for agent $i$; then we output any committee that includes $T_i \cup T_j$. However, the analysis of the distortion upper bound of our rule is intricate. 

\begin{algorithm}[ht]
\SetAlgoLined
Sort the agents in $N = [n]$ in a non-decreasing order of their cost for $O$ such that $c_i(O) \le c_j(O)$ for all $i < j$ \;
$S \leftarrow \varnothing$\;
\For{$i=1,\ldots,n$}{
    \If{$\top_i(O) \cap \top_j(O) = \varnothing$ for all $j \in S$ }{
        $S \leftarrow S \cup \{i\}$\;
    }
}
\caption{Construction of $S$ in Lemma~\ref{lem:best-of-some-voters}}
\label{alg:S-construction}
\end{algorithm}

\begin{algorithm}[t]
\SetAlgoLined
Choose an arbitrary agent $i \in N$\;
Choose an agent $j \in \arg\max_{\ell \in N \setminus \{i\}} c_i( T_\ell )$\;
Output an arbitrary committee $W \supseteq T_i \cup T_j$\;
\caption{\sc PolarOpposites}
\label{alg:2-agentFavorite}
\end{algorithm}

Before we proceed with bounding the distortion of our rule, we present a structural lemma, which holds for all possible values of $k$ and $q$, and will be extremely useful in the proof of the bound.  

\begin{lemma}\label{lem:best-of-some-voters}
Consider any instance $I=(N,A,d,k,q)$ and let $O \in \argmin_{C: |C|=k} \SC(C)$ be an optimal committee.
There exists a subset of agents $S$ with $|S| \le \floor{k/q}$ such that for every agent $i \in N$ there exists an agent $j \in S$ with $\top_i(O) \cap \top_j(O) \neq \varnothing$ and $c_j(O) \le c_i(O)$. In addition, for every agent $i \in N$ and committee $C \supseteq \bigcup_{j \in S}\ \top_j(A)$, we have $c_i(C) \le 3 \cdot c_i(O)$. 
\end{lemma}

\begin{proof}
We construct the set $S$ using Algorithm~\ref{alg:S-construction}. Note that because we are only interested in arguing the existence of $S$, we can assume access to the underlying costs in order to construct $S$. 

By construction, $S$ is such that for every agent $i \in N$ either $i \in S$ (in which case the condition in the statement of the lemma holds for $j=i$), or there exists $j \in N \setminus \{i\}$ that is considered before $i$ in Algorithm~\ref{alg:S-construction} (i.e., $c_j(O) \leq c_i(O)$) and $\top_i(O) \cap \top_j(O) \neq \varnothing$. So, it remains to show that $|S| \leq \lfloor k/q \rfloor$. 

Observe at any point during the algorithm, we maintain the invariant $\top_i(O) \cap \top_j(O) = \varnothing$ for all distinct $i,j \in S$. Since $|O| = k$ and $\top_i(O)$ for each $i \in S$ consists of $q$ unique alternatives from $O$, we always have $|S| \le \floor{k/q}$. 

For the second claim, consider any committee $C \supseteq \bigcup_{i \in S}\ \top_i(A)$. Clearly, $c_i(C) \leq c_i(O)$ for every $i \in S$. By the property of $S$ established above, for any agent $i \in N \setminus S$, there exist $j \in S$ and $x \in C$ such that $c_j(O) \le c_i(O)$ and $x \in \top_i(O) \cap \top_j(O)$. Let $y \in C$ be the $q$-th favorite alternative of $i$ in $\top_j(A)$. We make the following simple observations:
\begin{itemize}
    \item Since $x \in \top_i(O)$, it holds that $d(i,x) \leq c_i(O)$. 
    \item By the choice of $j$, and since $x \in \top_j(O)$ and $y \in \top_j(A)$, $d(j,x) \leq c_j(O) \leq c_i(O)$ and $d(j,y) \leq c_j(O) \leq c_i(O)$.
\end{itemize}
By the triangle inequality and the above observations, we have 
$$c_i(C) \leq d(i,y) \leq d(i,x) + d(j,x) + d(j,y) \leq 3 \cdot c_i(O).$$
This completes the proof.
\end{proof}

We are now ready to bound the distortion of the {\sc PolarOpposites} rule. 

\begin{theorem} \label{thm:k/3<=q<=k/2-upper}
For any $q \in (k/3,k/2]$, the $(k,q)$-distortion of {\sc PolarOpposites} is $O(n)$. 
\end{theorem}

\begin{proof}
Let $I=(N,A,d,k,q)$ be an instance with $q \in (k/3, k/2]$. Let $i$ and $j$ be the agents chosen by {\sc PolarOpposites} on $I$, $W \supseteq T_i \cup T_j$ be the committee returned by it, and $O \in \argmin_{C: |C|=k} \SC(C)$ be an optimal committee for $I$. We will show that $c_\ell(W) \leq c_\ell(O) + 4 \cdot \SC(O)$ for every agent $\ell \in N$. Then, by summing over all agents, we will obtain that $\SC(W) \le (4n+1) \cdot \SC(O)$, thus implying an upper bound of $4n+1$ on the distortion of {\sc PolarOpposites}. 
We distinguish between the following two cases.

\paragraph{Case 1: $\top_i(O) \cap \top_j(O) \neq \varnothing$.} Let $x \in \top_i(O) \cap \top_j(O)$ be an alternative that both $i$ and $j$ consider to be among their top $q$ alternatives in the optimal committee $O$. For any agent $\ell \in N$, let 
\begin{itemize}
    \item $y_\ell = \arg\max_{z \in T_\ell} d(\ell,z)$ be the $q$-th overall favorite alternative of $\ell$, and
    \item $y_{\ell j} = \arg\max_{z \in T_j} d(\ell,z)$ be the $q$-th favorite alternative of $\ell$ among those in $T_j$.
\end{itemize}
Since $T_j \subset W$, by the definition of $y_{\ell j}$, we have 
$c_\ell(W) \leq c_\ell(T_j) = d(\ell,y_{\ell j})$. 
Combining this with the triangle inequality, we obtain 
\begin{align*}
c_\ell(W) 
&\le d(\ell,y_{\ell j}) \\
&\leq d(\ell, y_\ell) + d(i,y_\ell) + d(i,x) + d(j,x) + d(j,y_{\ell j}).
\end{align*}
For any agent $v$, note that $c_v(T_v) \leq c_v(O)$, and
for any agent $u$ and alternative $y \in T_v$, we also have $d(u,y) \leq c_u (T_v)$. Combining these with the definitions of $j$, $y_\ell$, $x$, and $y_{\ell j}$, we can now bound each of the terms in the expression above as follows:
\begin{itemize}
    \item $d(\ell,y_\ell) = c_\ell(T_\ell) \leq c_\ell(O)$;
    \item $d(i,y_\ell) \leq c_i( T_\ell ) \leq c_i( T_j )$;
    \item $d(i,x) \leq c_i(O)$;
    \item $d(j,x) \leq c_j(O)$;
    \item $d(j, y_{\ell j}) \leq c_j(T_j) \leq c_j(O)$.
\end{itemize}
Substituting these, we get
\begin{align*}
c_\ell(W) \leq c_\ell(O) + c_i(T_j) + c_i(O) + 2c_j(O).
\end{align*}
We can bound the term $c_i(T_j)$ using the definition of alternative $y_{ij}$, the triangle inequality, and some of the above observations, as follows:
\begin{align*}
c_i(T_j) 
&= d(i, y_{ij}) \\
&\leq d(i,x) + d(j,x) + d(j, y_{ij}) \\
&\leq c_i(O) + c_j(O) + c_j(T_j) \\
&\leq c_i(O) + 2 c_j(O).
\end{align*}
So, combining everything, and using the fact that $\SC(O) \geq c_i(O) + c_j(O)$, we have that 
\begin{align*} 
c_\ell(W) 
&\leq c_\ell(O) + 2c_i(O) + 4c_j(O) \\
&\leq c_\ell(O) + 4 \cdot \SC(O),
\end{align*}
as desired.

\paragraph{Case 2: $\top_i(O) \cap \top_j(O) = \varnothing$.} 
Consider the set $S$ that is guaranteed to exist by Lemma~\ref{lem:best-of-some-voters}. Since $q \in (k/3, k/2]$, we have that $\lfloor k/q \rfloor = 2$, and hence $|S| \leq 2$.

If $|S|=1$, then every agent $\ell \in N$ is mapped to the single agent $v \in S$, and it holds that $\top_\ell(O) \cap \top_v(O) \neq \varnothing$. 
If $|S|=2$, we claim that there is a function $g: N \to S$ such that for every agent $\ell \in N$, it holds that $\top_\ell(O) \cap \top_{g(\ell)}(O) \neq \varnothing$ and $S=\set{g(i),g(j)}$. Lemma~\ref{lem:best-of-some-voters} already guarantees a function $g$ meeting the first condition. The only case in which the second condition cannot be met is if there exists $v \in S$ such that $\top_i(O) \cap \top_v(O) = \top_j(O) \cap \top_v(O) = \varnothing$. However, this, along with $\top_i(O) \cap \top_j(O) = \varnothing$ implies that $\top_i(O)$, $\top_j(O)$, and $\top_v(O)$ are disjoint subsets of $O$ of size $q$ each. This is a contradiction since $|O| = k < 3q$. 

Now, consider any agent $\ell \in N$, and suppose that $g(\ell) = g(j) = v \in S$; the case $g(\ell)=g(i)$ if $|S|=2$ is similar. 
By the properties of $S$, there exist alternatives $x \in \top_\ell(O) \cap \top_v(O)$ and $z \in \top_j(O) \cap \top_v(O)$. 
As in case 1, let $y_{\ell j} \in T_j$ be the $q$-th favorite alternative of $\ell$ in $T_j$; so, $c_\ell(T_j) = d(\ell, y_{\ell j})$. 
By the fact that $T_j \subset W$, and using the triangle inequality, we obtain
\begin{align*}
c_\ell(W) &\leq c_\ell(T_j) 
= d(\ell, y_{\ell j}) \\ 
&\le d(\ell,x) + d(v,x) + d(v,z) + d(j,z) + d(j,y_{\ell j}).
\end{align*}
By the definitions of $x$, $z$ and $y_{\ell j}$, we have
$d(\ell,x) \leq c_\ell(O)$,
$d(v,x) \leq c_v(O)$,
$d(v,z) \leq c_v(O)$,
$d(j,z) \leq c_j(O)$, and
$d(j, y_{\ell j}) \leq c_j(T_j) \leq c_j(O)$.
Combined with the fact that $\SC(O) \geq c_v(O)$ and $\SC(O) \ge c_j(O)$, we get 
\begin{align*}
c_\ell(W) 
&\leq c_\ell(O) + 2c_v(O) + 2c_j(O) \\
&\le c_\ell(O) + 4\cdot \SC(O),
\end{align*}
as desired.
\end{proof}

We conclude this section by showing a matching lower bound of $\Omega(n)$ on the distortion of any (even randomized) multiwinner voting rule. Hence, when $q \in (k/3, k/2]$,  {\sc PolarOpposites} is the asymptotically best possible rule in terms of distortion, even among randomized rules. 

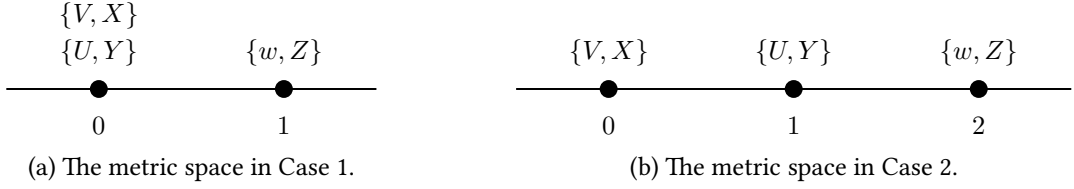
\begin{figure}[t]
\tikzset{every picture/.style={line width=0.75pt}} 
\begin{subfigure}[t]{0.45\linewidth}
\centering
\begin{tikzpicture}[x=0.7pt,y=0.7pt,yscale=-1,xscale=1]
\draw [line width=0.75]  (0,0) -- (200,0) ;
\filldraw (50,0) circle (3pt);
\filldraw (150,0) circle (3pt);

\draw (50,-40) node [inner sep=0.75pt]  [font=\small]  {$\{V, X\}$};
\draw (50,-20) node [inner sep=0.75pt]  [font=\small]  {$\{U, Y\}$};
\draw (50,20) node [inner sep=0.75pt]  [font=\small]  {$0$};

\draw (150,-20) node [inner sep=0.75pt]  [font=\small]  {$\{w, Z\}$};
\draw (150,20) node [inner sep=0.75pt]  [font=\small]  {$1$};
\end{tikzpicture}
\caption{The metric space in Case 1.}
\end{subfigure}
\ \ \ \ \
\begin{subfigure}[t]{0.45\linewidth}
\centering
\begin{tikzpicture}[x=0.7pt,y=0.7pt,yscale=-1,xscale=1]
\draw [line width=0.75]  (0,0) -- (300,0) ;
\filldraw (50,0) circle (3pt);
\filldraw (150,0) circle (3pt);
\filldraw (250,0) circle (3pt);

\draw (50,-20) node [inner sep=0.75pt]  [font=\small]  {$\{V, X\}$};
\draw (50,20) node [inner sep=0.75pt]  [font=\small]  {$0$};

\draw (150,-20) node [inner sep=0.75pt]  [font=\small]  {$\{U, Y\}$};
\draw (150,20) node [inner sep=0.75pt]  [font=\small]  {$1$};

\draw (250,-20) node [inner sep=0.75pt]  [font=\small]  {$\{w, Z\}$};
\draw (250,20) node [inner sep=0.75pt]  [font=\small]  {$2$};
\end{tikzpicture}
\caption{The metric space in Case 2.}
\end{subfigure}
\caption{The metric spaces used in the proof Theorem~\ref{thm:k/3<q-LB}. Each of the sets $X,Y,Z$ has size $q$. Both metrics are consistent to the ordinal profile according to which the preference of the $x$ agents in $V$ is $X \succ Y \succ Z$, the preference of all $x$ agents in $U$ is $Y \succ X \succ Z$, and the preference of agent $w$ is $Z \succ Y \succ X$. If the $q$ alternatives of $Z$ are not all included in the committee with positive probability (Case 1), then agent $w$ has positive expected $q$-cost in the first metric space, leading to unbounded distortion as the committee that includes $Z$ and $q$ alternatives from $X\cup Y$ has $q$-social cost $0$. If all alternatives of $Z$ are included in the committee (Case 2), then the $x$ agents of $V$ (in case not all alternatives of $X$ are included in the committee) or the $x$ agents in $U$ (in case not all alternatives of $Y$ are included in the committee) have $q$-cost $1$ in the second metric space, leading to a distortion of $x = \Omega(n)$ as the committee that includes all the alternatives of $X \cup Y$ and some alternatives of $Z$ has $q$-social cost $1$. In any case, the distortion of the voting rule is $\Omega(n)$.
}
\label{fig:k/3<q-LB}
\end{figure}

\begin{theorem} \label{thm:k/3<q-LB}
For $q \in (k/3, k/2]$, the $(k,q)$-distortion of every (even randomized) multiwinner voting rule is $\Omega(n)$.
\end{theorem}

\begin{proof}
Fix $k$ and $q \in (k/3,k/2]$, and let $f$ be an arbitrary multiwinner voting rule. We consider instances with $n = 2x+1$ agents, partitioned into two sets $V$ and $U$ of size $x$ each and a singleton $\set{w}$. There are $m=3q$ alternatives, partitioned into three sets $X$, $Y$, and $Z$ consisting of $q$ alternatives each. 

Consider any preference profile subject to the following rules.
\begin{itemize}
\item 
Every agent ranks the alternatives that belong to the same set arbitrarily.

\item
Every agent in $V$ has the ranking $X \succ Y \succ Z$.

\item 
Every agent in $U$ has the ranking $Y \succ X \succ Z$. 

\item 
Agent $w$ has the ranking $Z \succ Y \succ X$.
\end{itemize}
Since $m=3q > k$, the committee returned by $f$ cannot include every alternative with probability $1$. We switch between the following two cases. Figure~\ref{fig:k/3<q-LB} depicts the two metric spaces considered in these cases. 

\paragraph{Case 1: With a positive probability, the committee chosen by $f$ does not include all $q$ alternatives of $Z$.} \ \\
Consider the following metric, which is consistent (up to tie-breaking) with the preference profile defined above:
\begin{itemize}
\item 
The agents in $V \cup U$ and the alternatives in $X\cup Y$ are all located at $0$.

\item 
Agent $w$ and the alternatives of $Z$ are located at $1$. 
\end{itemize}

Since some alternative in $Z$ is not included in the chosen committee with a positive probability, the expected $q$-cost of agent $w$, and thus the expected $q$-social cost under $f$, is positive. However, the committee that includes all of $Z$ and any $q$ alternatives from $X \cup Y$ has $q$-social cost $0$. Therefore, the distortion of $f$ is unbounded. 

\paragraph{Case 2: With probability $1$, the committee chosen by $f$ includes every alternative of $Z$.} \ \\
Consider the following metric, which is again consistent (up to tie-breaking) with the preference profile defined above:
\begin{itemize}
\item
The agents in $V$ and the alternatives in $X$ are at $0$.

\item 
The agents in $U$ and the alternatives in $Y$ are at $1$.

\item 
Agent $w$ and the alternatives in $Z$ are  at $2$.
\end{itemize}
Recall that with probability $1$, $f$ chooses a committee that includes all $q$ alternatives of $Z$. Any such committee cannot contain all alternatives in $X \cup Y$. If it does not include an alternative of $X$, the $q$-cost of every agent in $V$ is at least $1$, and if it does not include an alternative of $Y$, the $q$-cost of every agent in $U$ is at least $1$. Either way, we have that, with probability $1$, $f$ chooses a committee with $q$-social cost at least $x$. Hence, the expected $q$-social cost under $f$ is at least $x$. 
However, the committee that includes all $2q$ alternatives of $X \cup Y$ and  $k-2q$ arbitrary alternatives of $Z$ has $q$-social cost $1$; only agent $w$ has $q$-cost $1$. So, the distortion of $f$ is at least $x = \Omega(n)$ in this case. 
\end{proof}

\section{Constant Distortion With $q > k/2$}
We now turn our attention to the case of $q > k/2$, where it is possible to achieve constant distortion. A crucial observation, which we exploit in this section, is that the $q$-costs of the agents form a new metric over the agents and all possible committees of alternatives. 

\begin{lemma} \label{lem:q>k/2-metric}
For any instance $I=(N,A,d,k,q)$ with $q > k/2$, the $q$-costs $c_i(C)$ of agents $i \in N$ for $k$-sized committees of alternatives $C \subset A$ form a pseudometric.
\end{lemma}

\begin{proof}
To prove the statement, we need to show that the $q$-costs satisfy the triangle inequality, i.e., $c_i(X) \leq c_i(Y)+c_j(Y)+c_j(X)$, for any two agents $i,j \in N$ and two $k$-sized committees $X,Y$. 

For any $q \leq k$, there exists $x \in \top_j(X)$ that is among the $k-q+1$ {\em least favorite} alternatives (ranked in some position from $q$ to $k$) of agent $i$ in $X$; thus, $x$ is such that $c_i(X) \leq d(i,x)$ and $c_j(X) \geq d(j,x)$. Since $q>k/2$, there also exists $y \in \top_i(Y) \cap \top_j(Y)$, that is, $y$ is such that $c_i(Y) \geq d(i,y)$ and $c_j(Y) \geq d(j,y)$. Combining these with the triangle inequality for the distances between agents and alternatives, we have that
\begin{align*}
c_i(X) \leq d(i,x) &\leq d(i,y) + d(j,y)+d(j,x) \\
&\leq c_i(Y)+c_j(Y)+c_j(X),
\end{align*}
as desired. 
\end{proof}

Since the $q$-costs form a metric, we can transform the original profile in which the agents rank the alternatives to one in which the agents rank all possible $k$-sized committees. In particular, to decide whether agent $i$ prefers a committee $X$ over another committee $Y$, it suffices to compare her $q$-th favorite alternatives in $X$ and $Y$; we can break ties arbitrarily. 
Given the rankings of the agents over the committees, we can then employ any single-winner rule to decide the final committee. Specifically, using the {\sc Plurality-Matching} rule of \citet{gkatzelis2020resolving}, we obtain a best-possible distortion bound of $3$. 

\begin{corollary}
For $q > k/2$, there exists a multiwinner voting rule with $(k,q)$-distortion at most $3$.
\end{corollary}

Unfortunately, the above approach requires us to apply a single-winner voting rule to a profile consisting of an exponential number of alternatives (the set of all $k$-sized committees). This is an inherent obstacle in our attempts to get constant distortion bounds by na\"{\i}vely applying known deterministic single-winner voting rules.

Interestingly, it is still easy to compute the favorite $k$-sized committee of an agent, which consists of the agent's $k$ most preferred alternatives; we refer to this as the {\em top-$k$} committee of the agent. Consequently, randomized dictatorship, the single-winner voting rule which selects an agent uniformly at random and returns her most favorite alternative, can be efficiently implemented in our setting. Using its distortion analysis by~\citet{AP17}, we obtain the following.

\begin{corollary}
For any $q>k/2$, the $(k,q)$-distortion of randomized dictatorship, which can be implemented efficiently, is at most $3-2/n$.
\end{corollary}

In the following, we restrict our attention to deterministic polynomial-time multiwinner voting rules. Our main result provides such a rule with distortion at most $9$ by exploiting the following lemma.

\begin{lemma} \label{lem:top-k-optimal}
Consider any instance $I=(N,A,d,k,q)$ with $q > k/2$, and let $O$ be an optimal committee for $I$. 
There exists a committee $C$ that is top-$k$ for some agent, and such that $\SC(C) \leq 3 \cdot \SC(O)$.
\end{lemma}

\begin{proof}
By Lemma~\ref{lem:q>k/2-metric}, the $q$-costs of the agents form a metric space. 
Let $j$ be the closest agent to the optimal committee $O$ according to her $q$-cost, i.e., $c_j(O) \leq c_i(O)$ for every $i \in N$. 
Let $C$ be the top-$k$ committee of $j$, and thus $c_j(C)\leq c_j(O)$.
Combining the above with the triangle inequality for the $q$-cost metric, for any agent $i \in N$, we have
\begin{align*}
c_i(C) 
&\leq c_i(O) + c_j(O) + c_j(C) \\
&\leq c_i(O) + 2c_j(O) \leq 3 c_i(O),
\end{align*}
The lemma follows by summing over all agents. 
\end{proof}

Essentially, Lemma~\ref{lem:top-k-optimal} suggests that the best top-$k$ committee must $3$-approximate the optimal committee in terms of social cost. So, by considering the $n$ top-$k$ committees (one per agent) as alternatives, we can deploy the single-winner rule of~\citet{gkatzelis2020resolving} to obtain in polynomial time a committee that is within a factor of $3$ away from the best top-$k$ committee in terms of social cost, and thus has a $(k,q)$-distortion of at most $9$.

\begin{corollary}\label{cor:dist-9}
For any $q > k/2$, there is a polynomial-time multiwinner voting rule with $(k,q)$-distortion at most $9$.
\end{corollary}

Actually, the approach used to obtain Corollary~\ref{cor:dist-9} can be used as a template that could potentially lead to even lower distortion bounds by deterministic polynomial-time voting rules. All we need is an algorithm 
that takes as input the ranking profile and decides --- in polynomial time --- on a set $P$ of $k$-sized committees, which, for every distance metric consistent with the ranking profile, contains a committee that approximates the optimal social cost within a factor of $\rho$ (e.g., Lemma~\ref{lem:top-k-optimal} suggests that this is possible for $\rho=3$). Then, applying the voting rule of~\citet{gkatzelis2020resolving} on the reduced  ranking profile with only the committees in $P$ as alternatives, we obtain a deterministic polynomial-time multiwinner voting rule with $(k,q)$-distortion at most $3\rho$.

We present two results related to this template. The first one is positive and shows that a guarantee of $\rho=1$ is possible when the number of agents is constant, making the optimal distortion bound of $3$ achievable in polynomial time. 

\begin{theorem} \label{thm:q>k/2-constant-n}
For any $q>k/2$ and constant number of agents, there is a deterministic polynomial-time multiwinner voting rule with $(k,q)$-distortion at most $3$.
\end{theorem}

\begin{proof}
We prove the theorem by defining a voting rule that follows our template with $\rho=1$. In particular, given a ranking profile, our rule first identifies a set $P$ of at most $m^n$ committees such that one of minimum social cost is always included, no matter which distance function $d$, consistent with the preference profile, is used in the definition of the social cost. Then, by invoking the  single-winner rule of \citet{gkatzelis2020resolving} using the committees in $P$ as alternatives, we get the desired multiwinner voting rule with $(k,q)$-distortion of at most $3$.

To identify the committees in $P$, we enumerate over all possibilities for the alternatives in a $k$-sized committee that the agents can have as their $q$-th closest ones. In particular, for each of the $m^n$ possible vectors of alternatives $\bell=\langle \ell_1, \ell_2, ..., \ell_n\rangle$, we include in $P$ a committee $C\subseteq A$ with $|C|=k$, $\bigcup_{i=1}^n{\{\ell_i\}}\subseteq C$, such that $\ell_i$ is the alternative in $C$ that is the $q$-th closest to agent $i$, for $i \in [n]$, if such a committee exists. We refer to each committee having these properties as a {\em $(k,q)$-completion} of the vector of alternatives $\bell$. Notice that including in $P$ just one $(k,q)$-completion for each vector of alternatives is enough for our purposes since any two committees $C$, $C'$ that are $(k,q)$-completions for the same vector $\mathbf{\ell}$ of alternatives have the same social cost for a given distance function $d$:
$\SC_q(C|d) = \SC_q(C'|d) =\sum_{i=1}^n{d(i,\ell_i)}.$
Clearly, since $n$ is a constant, $P$ has polynomial size. To prove that its whole construction takes polynomial time, it suffices to show how a $(k,q)$-completion $C$ for a given vector of alternatives $\bell$ can be identified (if it exists) in polynomial time. Essentially, we need to decide which alternatives in addition to the ones in $\bell$ form a $k$-sized committee and ensure that alternative $\ell_i$ is the $q$-th closest to agent $i$ among those in $C$. To do so, we define the following classification of the alternatives into types from the set 
\begin{align*}
    \types &= \{\langle t_1, t_2, ..., t_n\rangle: t_i\in\{+1,0,-1\}, i\in N\}.
\end{align*}
An alternative $a\in A$ belongs to type $t\in \types$ if and only if
\begin{itemize}
    \item $a \succ_i \ell_i$ for each $i\in N$ with $t_i=+1$,
    \item $a=\ell_i$ for each $i\in N$ with $t_i=0$, and
    \item $\ell_i \succ_i a$ for each $i\in N$ with $t_i=-1$.
\end{itemize}
Notice that the classification is such that replacing an alternative in a $(k,q)$-completion by another alternative from the same class results in another $(k,q)$-completion. Hence, to identify a $(k,q)$-completion $C$ for a given vector of alternatives $\bell$, we need to identify the number of alternatives from each type of $\types$ that $C$ should have.

For each type $t\in \types$, denote by $H(t)$ the number of alternatives of type $t$. Also, set $L(t)=1$ if $t$ is the type of some alternative in $\bell$ and $L(t)=0$, otherwise. The quantities $L(t)$ and $H(t)$ are lower and upper bounds on the number $h(t)$ of alternatives of type $t$ that can be included in a $(k,q)$-completion of $\bell$. In particular, notice that each alternative in $\bell$ is the unique alternative in its type $t$, i.e., $H(t)=1$. Setting $L(t)=1$ for this type guarantees that this alternative will be included in the $(k,q)$-completion. Now, the existence of a $(k,q)$-completion is equivalent to the feasibility of the following integer linear program:
\begin{align*}
    & \sum_{t\in \types}{h(t)}=k\\
    & \sum_{t\in \types: t_i\geq 0}{h(t)}=q, \forall i\in N\\
    & L(t) \leq h(t) \leq H(t), \forall t\in \types\\
    & h(t)\in \N_{\geq 0}, \forall t\in \types
\end{align*}
Notice that we can naively check the feasibility of the above ILP by enumerating all possible values for the variables $h(t)$. As there are $3^n$ types, there are $3^n$ such variables (i.e., constantly many since $n$ is a constant), each taking at most $m+1$ values. Even this naive solution takes only polynomial time. Finally, once we have a feasible solution for the above ILP, we can easily define the $(k,q)$-completion by just including $h(t)$ alternatives of type $t$ in it, for $t\in \types$.
\end{proof}

Our last result is negative and reveals a limitation of our template. It shows a lower bound of $1+2/e$ on $\rho$ (unless $\text{P} = \text{NP}$) and, consequently, a lower bound of approximately $5.207$ on the distortion that can be achieved via this template. In the appendix, we present a slightly weaker lower bound of $4.5$ that is purely information-theoretic and does not rely on any complexity assumption. 

\begin{theorem}\label{thm:hard-P}
Let $\mathcal{A}$ be an algorithm which on input a ranking profile with $n$ agents and $m$ alternatives, and integers $k$ and $q$ with $1\leq k/2 < q \leq k < m$, runs in time polynomial in $m$ and $n$ and returns a set $P$ of (polynomially many) $k$-sized committees of alternatives with the following property: For every distance function $d$, there is a committee $C\in P$ such that $\SC_q(C|d)< (1+2/e+\varepsilon) \cdot \SC_q(O|d)$, where $O$ is a committee of minimum social cost according to $d$, and $\varepsilon$ is a positive constant. Then, $\text{P} = \text{NP}$.
\end{theorem}

\begin{proof}
We will prove the theorem using a reduction from the maximum $K$-cover problem. In maximum $K$-cover, we are given a universe $U$ of elements, a collection $\mathcal{S}$ of subsets of $U$, and an integer $K$. The objective is to find a subcollection of $K$ sets from $\mathcal{S}$ (or, a {\em $K$-collection}), whose union contains the maximum number of elements from $U$.

A well-known hardness result due to~\citet{Fei98} states that, for every constant $\varepsilon>0$, it is NP-hard to distinguish between the next two types for a maximum $K$-cover instance:
\begin{enumerate}
\item The instance has a $K$-collection that covers all elements of $U$.
\item Any $K$-collection covers at most a $(1-1/e-\varepsilon/2)$-fraction of the elements of $U$.
\end{enumerate}

Our reduction works as follows: Given an instance of maximum $K$-cover, we build a ranking profile with $n=|U|$ agents and $m=q-1+|\mathcal{S}|$ alternatives. Each agent corresponds to an element of $U$. There are $q-1$ special alternatives and, additionally, a distinct set-alternative for every set $S\in \mathcal{S}$. For $i\in U$, agent $i$ ranks the special alternatives first (in arbitrary order), followed by the set-alternatives that include element $i$, and then the remaining set-alternatives. The distance function $d$ is such that agent $i$ has distance $1$ from the special alternatives and the set-alternatives which contain element $i$, and distance $3$ from the remaining set-alternatives. The committee size is $k=q-1+K$.

Let us call {\em canonical} a $k$-sized committee that includes all special alternatives. Notice that a committee can be transformed to a canonical one by including the missing special alternatives and removing (arbitrarily) an equal number of set-alternatives. This transformation guarantees that the social cost (according to $d$) of the canonical committee is not worse than that of the original one. 

We now claim that if algorithm $\mathcal{A}$ satisfies the condition in the statement of Theorem~\ref{thm:hard-P}, then we can distinguish between the two cases above in polynomial time, contradicting the hardness statement of~\citet{Fei98}.
Indeed, consider the algorithm which, on input a ranking profile produced by our reduction, runs algorithm $\mathcal{A}$ first to compute a set $P$ of (polynomially many) $k$-sized committees, then transforms these committees to canonical ones, and, among the $K$-collections corresponding to these canonical committees (each formed by the sets that correspond to the set-alternatives), returns the one that covers the maximum number of elements.

Clearly, any $K$-collection returned by the algorithm covers at most a $(1-1/e-\varepsilon/2)$-fraction of the elements if the maximum $K$-cover instance is of type 2. If it is of type 1, let $\widehat{O}$ be a $K$-collection that covers all elements of $U$. Then, the canonical committee $O$ that contains the special alternatives and the set-alternatives corresponding to the sets of $\widehat{O}$ has a minimum possible social cost of $\SC_q(O|d)=n$. Indeed, since some of the sets of $\widehat{O}$ contain element $i$, agent $i$ has its $q$-th favorite alternative in $O$ at distance $1$ (recall that she has the special alternatives in the top $q-1$ positions of its preference, so the $q$-th closest alternatives is her top set-alternative in $O$). Now, by the approximation guarantee of algorithm $\mathcal{A}$, we know that there is a canonical committee among those obtained by the committees of $P$, that has $\SC_q(C|d) < (1+2/e+\varepsilon)n$. This means that the number of agents whose $q$-th closest alternative in $C$ is at distance $1$ is more than $(1-1/e+\varepsilon/2)n$ (the rest are at distance $3$). Equivalently, by the definition of our reduction, the $K$-collection that corresponds to $C$ covers more than $(1-1/e+\varepsilon/2)n$ elements of $U$. Hence, detecting whether the original instance of maximum $K$-cover is of type 1 or 2 can be decided by comparing the number of elements covered by the outcome of the algorithm with $(1-1/e+\varepsilon/2)n$.
\end{proof}

\section{Extensions and Open Problems}\label{sec:discussion}

In this work, we extended the metric distortion framework to multiwinner voting. By modeling the cost of an agent for a $k$-sized committee as her distance from her $q$-favorite alternative therein, we revealed a quite surprising trichotomy on the distortion of multiwinner rules in terms of $q$ and $k$, providing asymptotically tight bounds. The main question that our work leaves open is to identify the best possible distortion for when $q > k/2$ that can be achieved by an efficient deterministic rule.

We exclusively focused on the social cost and did not consider other objectives, such as the maximum agent cost. It is not hard to observe that our methods provide bounds in terms of this objective as well, albeit not necessarily tight. When $q \leq k/3$, the distortion remains unbounded (in the instances used in the proof of Theorem~\ref{thm:q<=k/3-LB}, the optimal committee guarantees cost $0$ to all agents, whereas any voting rule yields a positive cost to some agent). When $q \in (k/3,k/2)$, by carefully inspecting the proof of the upper bound of {\sc PolarOpposites} in Theorem~\ref{thm:k/3<=q<=k/2-upper}, we can show that it achieves constant distortion in terms of the maximum cost. Finally, when $q > k/2$, since the $q$-costs form a metric space, known results from single-winner voting extend to multiwinner voting; that is, there exists a deterministic multiwinner rule with distortion at most $3$, which is based on the rule of \citet{gkatzelis2020resolving}; this rule is known to achieve a \emph{fairness ratio} of at most $3$, which implies $3$-approximation of the maximum cost as well. The efficient upper bound of $9$ via our template also extends to the maximum cost, as the bound of $3$ in Lemma~\ref{lem:top-k-optimal} actually holds per agent. In the future, it would be interesting to go beyond objectives such as the social and maximum costs, and consider others that make sense in mutliwinner voting.

\section*{Acknowledgements}
We would like to thank Elliot Anshelevich for fruitful discussions in early stages of this work, and Piotr Skowron for suggesting an example demonstrating unbounded distortion for $q=1$, which has been generalized to $q \le k/3$ in \Cref{thm:q<=k/3-LB}.

\bibliographystyle{plainnat}
\bibliography{abb,ultimate,distortion,references}

\section*{Appendix}



\section*{A Weaker Information-Theoretic Lower Bound}
\begin{theorem}
Let $\epsilon>0$ be a constant, $m\geq 2$ an even integer, and $q=k=m/2$. There exist a set $A$ of $m$ alternatives, a set $N$ of $m$ agent with a preference profile $(\succ_i)_{i\in N}$ over the alternatives of $A$, a distance function $d$ that is consistent with the preference profile $\succ$, and a $k$-sized committee $O$, such that for any set $P$ of less than $\exp(\epsilon^2 m)$ $k$-sized committees of alternatives, it holds $\SC_q(C) \geq \left(\frac{3}{2}-\epsilon\right) \SC_q(O)$ for every committee $C\in P$.
\end{theorem}

\begin{proof}
Let $A=\{a_1, ...a_m\}$, $N=\{1, 2, ..., m\}$, and consider the following preference profile $(\succ_i)_{i\in N}$: For $i=1, 2, ..., m/2$, both agents $2i-1$ and $2i$ have alternatives $a_{2i-1}$ and $a_{2i}$ at the two bottom positions of their ranking, while the remaining alternatives are distributed arbitrarily in the remaining $m-2$ positions. Agent $2i-1$ has alternative $a_{2i-1}$ last (i.e., $a_{2i}\succ_{2i-1} a_{2i-1}$) while agent $2i$ has alternative $a_{2i}$ last (i.e., $a_{2i-1} \succ_{2i} a_{2i}$).

Our proof uses the probabilistic method (e.g., see~\cite{MR95}). We will define a random distance function $d$ that is consistent with the profile above, as well as a random committee $O$ of $k=m/2$ alternatives, so that by defining the social cost $\SC_q$ using $d$, we have $\SC_q(O)=m$ and, with strictly positive probability, $\SC_q(C)> \left(\frac{3}{2}-\epsilon\right)m$ for each committee $C\in P$. This probabilistic argument guarantees that there exists a distance function such that $\SC_q(C) > \left(\frac{3}{2}-\epsilon\right) \SC_q(O)$, for every committee $C\in P$.

Initially, committee $O$ is empty. For $i=1, 2, ..., m/2$, toss a fair coin, independently for different $i$'s:
\begin{itemize}
    \item On heads, set $d(2i-1,a_{2i-1})=3$, $d(2i-1,a_j)=1$ for $j\not=2i-1$, and $d(2i,a_j)=1$ for $j=1, ..., m$. Include alternative $a_{2i}$ in $O$.
    \item On tails, set $d(2i,a_{2i})=3$, $d(2i,a_j)=1$ for $j\not=2i$, and $d(2i-1,a_j)=1$ for $j=1, ..., m$. Include alternative $a_{2i-1}$ in $O$.
\end{itemize}
Observe that for each agent, all alternatives are at distance $1$, besides (possibly) the alternative that is ranked last, which is at distance $3$. Hence, $d$ is indeed a distance function that is consistent with the profile. In addition, the alternative that is included in committee $O$ at step $i$ is the $q$-th closest (recall that $q=m/2$) to both agents $2i-1$ and $2i$ and, furthermore, it is at distance $1$ from both of them. Hence, $\SC_q(O)=m$.

Now, consider a committee $C\in P$. For $i=1, ..., m/2$, let $X_i$ be the random variable denoting the sum $c_q(2i-1,C)+c_q(2i,C)$. Notice that
\begin{align}\label{eq:sum-of-irv}
    \SC_q(C) &= \sum_{i=1}^n{c_q(i,C)}=\sum_{i=1}^{m/2}{X_i}.
\end{align}
Now consider each random variable $X_i$ separately. We distinguish between four cases:
\begin{itemize}
    \item If $C$ contains both alternatives $a_{2i-1}$ and $a_{2i}$, then $c_q(2i-1,C)=d(2i-1,a_{2i-1})$ and $c_q(2i,a_{2i})=d(2i,a_{2i})$. Notice that, by the definition of the random process, among these two quantities, one is equal to $1$ and the other is equal to $3$. Hence, $X_i=4$.
    \item If $C$ contains neither $a_{2i-1}$ nor $a_{2i}$, then both agents $2i-1$ and $2i$ are at distance $1$ from all alternatives in $C$. Hence, $X_i=2$.
    \item If $C$ contains alternative $a_{2i-1}$ but not alternative $a_{2i}$, notice that $a_{2i-1}$ is the alternative of $C$ that is the $q$-th closest to both agents $2i-1$ and $2i$. Recall that $d(2i,a_{2i-1})$ is always equal to $1$ while $d(2i-1,a_{2i-1})$ is equiprobably equal to either $3$ (and, hence, $X_i=4$) or $1$ (and, hence, $X_i=2$), depending on whether the outcome of the $i$-th coin toss is heads or tails, respectively.
    \item The case when $C$ contains alternative $a_{2i}$ but not $a_{2i-1}$ is completely symmetric. Then, $X_i=2$ or $X_i=4$ equiprobably, depending on whether the outcome of the $i$-th coin toss is heads or tails, respectively.
\end{itemize}

We conclude that $X_1, ..., X_{m/2}$ are independent random variables, taking their values in $\{2,4\}$, with 
\begin{align}\label{eq:exp-of-x_i}
    \E[X_i] &= 2+|C \cap \{a_{2i-1},a_{2i}\}|.
\end{align}
Hence, we can use the well-known Hoeffding bound to estimate the probability $\Pr[\SC_q(C)\leq \left(\frac{3}{2}-\epsilon\right)m]$.

\begin{lemma}[\citealt{Hoeff63}]\label{lem:hoeffding}
Let $Z_1$, $Z_2$, ..., $Z_\ell$ be independent random variables with $Z_i\in [a_i,b_i]$ and $Z=\sum_{i=1}^\ell{Z_i}$. Then, for every $t> 0$,
\begin{align*}
    \Pr[\E[Z]-Z \geq t] &\leq \exp\left(-\frac{2t^2}{\sum_{i=1}^\ell{(b_i-a_i)^2}}\right).
\end{align*}
\end{lemma}

By (\ref{eq:sum-of-irv}), (\ref{eq:exp-of-x_i}), and using linearity of expectation, we have
\begin{align*}
    \E[\SC_q(C)] &= \sum_{i=1}^{m/2}{\E[X_i]} = \sum_{i=1}^{m/2}{\left(2+|C\cap \{a_{2i-1},a_{2i}\}|\right)} = 3m/2.
\end{align*}
Using Lemma~\ref{lem:hoeffding} with $\ell=m/2$, $Z=\SC_q(C)$, $Z_i=X_i$, $b_i=4$, and $a_i=2$ for $i=1, ..., m/2$, and $t=\epsilon m$, we get
\begin{align*}
    \Pr[\SC_q(C)\leq \left(\frac{3}{2}-\epsilon\right)m]= \Pr[\E[\SC_q(C)]-\SC_q(C)\geq \epsilon m] \leq \exp\left(-\epsilon^2 m\right).
\end{align*}
By applying the union bound, the probability that $\SC_q(C)\leq \left(\frac{3}{2}-\epsilon\right)m$ for some committee $C\in P$, is less than $1$. We conclude that, with strictly positive probability, $\SC_q(C)>\left(\frac{3}{2}-\epsilon\right)m$ for every committee $C\in P$, as desired.
\end{proof}

\end{document}